\documentclass[journal,10pt]{IEEEtran}
\usepackage{amsmath,graphicx}
\usepackage{dsfont}
\usepackage{float}
\usepackage{xfrac}
\usepackage{color}
\usepackage{psfrag}
\usepackage{amsthm}
\usepackage{subfigure}

\usepackage[lined,boxed,commentsnumbered]{algorithm2e}
\usepackage{enumerate}
\usepackage{import}
\usepackage{epsfig}
\usepackage{amssymb}
\usepackage{url}
\usepackage{subref}
\usepackage{trfsigns}

\DeclareMathAlphabet{\bm}{OML}{cmr}{bx}{it}
\DeclareMathAlphabet{\mathsf}{OT1}{cmss}{m}{n}
\DeclareMathAlphabet{\bs}{T1}{cmss}{bx}{sl}
\DeclareMathAlphabet{\ms}{T1}{cmss}{m}{sl}
\DeclareMathAlphabet{\mathpzc}{OML}{zplm}{m}{it}

\newcommand{\bg}[1]{\boldsymbol{#1}} 

\newcommand{\bb}{\mathbb}
\newcommand{\rs}{\mathrm}
\newcommand{\mc}{\mathcal}
\newcommand{\bh}[1]{\hat{\bm #1}}
\newcommand{\bt}[1]{\tilde{\bm #1}}
\newcommand{\ds}{\mathds}

\newtheorem{lemma}{Lemma}
\newtheorem{proposition}{Proposition}

\newtheorem{remark}{Remark}
\newtheorem{definition}{Definition}
\newtheorem{example}{Example}

\newtheorem{problem}{Problem}
\newtheorem{assumption}{Assumption}
\newtheorem{conjecture}{Conjecture}

\begin{document}

\title{Optimal deep neural networks for sparse recovery \\ via Laplace techniques}

\author{Steffen~Limmer 
        and~S\l awomir~Sta\'nczak,~\IEEEmembership{Senior~Member,~IEEE}
\thanks{This work was supported in part by the German Research Foundation (DFG) under Grant STA 864/8-1.}
\thanks{S. Limmer and S. Sta\'nczak are with the Network-Information Theory Group, TU Berlin.} 
}

\markboth{}
{Limmer \MakeLowercase{\textit{et al.}}: Optimal deep neural networks for sparse recovery \\ via Laplace techniques}

\maketitle

\begin{abstract}
This paper introduces Laplace techniques for designing a neural network, with the goal of estimating simplex-constraint sparse vectors from compressed measurements. To this end, we recast the problem of MMSE estimation (w.r.t. a pre-defined uniform input distribution) as the problem of computing the centroid of some polytope that results from the intersection of the simplex and an affine subspace determined by the measurements. Owing to the specific structure, it is shown that the centroid can be computed analytically by extending a recent result that facilitates the volume computation of polytopes via Laplace transformations. A main insight of this paper is that the desired volume and centroid computations can be performed by a classical deep neural network comprising threshold functions, rectified linear (ReLU) and rectified polynomial (ReP) activation functions. The proposed construction of a deep neural network for sparse recovery is completely analytic so that time-consuming training procedures are not necessary. Furthermore, we show that the number of layers in our construction is equal to the number of measurements which might enable novel low-latency sparse recovery algorithms for a larger class of signals than that assumed in this paper. To assess the applicability of the proposed uniform input distribution, we showcase the recovery performance on samples that are soft-classification vectors generated by two standard datasets. As both volume and centroid computation are known to be computationally hard, the network width grows exponentially in the worst-case. It can be, however, decreased by inducing sparse connectivity in the neural network via a well-suited basis of the affine subspace. Finally, the presented analytical construction may serve as a viable initialization to be further optimized and trained using particular input datasets at hand.
\end{abstract}

\begin{IEEEkeywords}
sparse recovery, compressed sensing, deep neural networks, convex geometry
\end{IEEEkeywords}

\IEEEpeerreviewmaketitle

\section{Introduction}
Deep neural networks have enabled significant improvements in a series of areas ranging from image classification to speech recognition \cite{GooBenCou16}.
While training neural networks is in general extremely time-consuming, the use of trained networks for real-time inference has been shown to be feasible, which has rendered successful commercial applications as self-driving cars and real-time machine translation possible \cite{GooBenCou16}. These findings took many researchers by surprise as the underlying problems were believed to be both computationally and theoretically hard. Yet, theoretical foundations of successful applications remain thin and analytical insights have only recently appeared in the literature. For instance, results on the translation and deformation stability of convolutional neural networks \cite{Mal16,WiaBoe15} deserve a particular attention. In this paper, we are interested in the interplay between the sparsity of input signals and the architecture of neural networks. This interplay between sparsity and particular network architectures has been so far only addressed in numerical studies \cite{XinWanGaoWip16,KamMan15,WanLinHua16} or the studies are based on shallow architectures \cite{LimSta16}. This paper contributes towards closing this important gap in theory by considering a simple sparse recovery problem that is shown to admit an optimal solution by an analytically-designed neural network.

More precisely, we address the problem of estimating a non-negative compressible vector from a set of noiseless measurements via an $M$-layer deep neural network. The real-valued input vector $\bm{x} \in \bb{R}_+^N$ is assumed to be distributed uniformly over the standard simplex and is to be estimated from $M<N$ linear measurements $\bm{y}\in \bb{R}^M$ given by 
\begin{align}\label{equ:equ_measurements}
\bm{y} = \bm{A} \bm{x}. 
\end{align}
Here and hereafter, the measurement matrix $\bm{A}\in\bb{R}^{M \times N}$ is assumed to be an arbitrary fixed full-rank matrix. We refer the reader to Fig. \ref{fig:realizations} in Sec. \ref{sec:numresults} for an illustration of a realization $\bm{x}$ of the stochastic process considered in this paper as well as markedly similar soft-classification vectors generated by standard datasets.
The problem of designing efficient algorithms for recovering $\bm{x}$ given $\bm{y}$ has been at the core of many research works with well-understood algorithmic methods, including $\ell_1$-minimization \cite{ForRau11} and iterative soft/hard-thresholding algorithms (ISTA / IHT) \cite{ForRau11,BecTeb09}. Training structurally similar neural networks has been investigated in \cite{XinWanGaoWip16,KamMan15,WanLinHua16}; in particular, these works focused on the problem of fine-tuning parameters using stochastic gradient descent because a full search over all possible architectures comprising varying number of layers, activation functions and size of the weight matrices is infeasible. In contrast to these approaches, this paper introduces multidimensional Laplace transform techniques to obtain both the network architecture as well as the parameters in a fully analytical manner without the need for data-driven optimization. Interestingly, this approach reveals an analytical  explanation for the effectiveness of threshold functions, rectified linear (ReLU) and rectified polynomial (ReP) activation functions. Moreover, in our view, the resulting Laplace neural network can be applied to a larger class of sparse recovery problems than that assumed in this paper which is supported by a numerical study.

\subsection{Notation and Some General Assumptions}\label{sec:not}

Throughout the paper, the sets of reals, nonnegative reals and reals excluding the origin are designated by $\bb{R}$, $\bb{R}_{+}$ and $\bb{R}_{\neq 0}$, respectively. $\bb{S}^{N-1}\subset\bb{R}^N$ and $\Delta\subset\bb{R}^N$ denote the $N$-dimensional unit sphere and the standard simplex defined as $\Delta:= \{\bm{x}\in\bb{R}_{+}^N: \sum\nolimits_{n=1}^N x_n \leq 1\}$. We use lowercase, bold lowercase and bold uppercase serif letters $x$, $\bm{x}$, $\bm{X}$ to denote scalars, vectors and matrices, respectively. Sans-serif letters are used to refer to scalar, vector and matrix random variables $\ms{x}$, $\bs{x}$, $\bs{X}$. 
Throughout the paper, all random vectors are functions defined over a suitable probability measure space $(\Omega,\mc{A},p)$ where $\Omega$ is a sample space, $\mc{A}$ is a $\sigma$-algebra, and $p:\mc{A}\to[0,1]$ a probability measure on $\mc{A}$.\footnote{In particular, if the random vectors are in $\bb{R}_+^N$, then $\Omega=\bb{R}_+^N$ and $\Delta\in\mc{A}$ with $\mc{A}$ being the power set of $\bb{R}_+^N$.}. We use $p(\mc{X})=p(\bm{x}\in\mc{X})$ where $\mc{X}\in\mc{A}$ and $\mc{U}(\mc{X})$ to denote the uniform distribution over the set $\mc{X}\in\mc{A}$; finally, functions of random vectors are assumed to be measurable functions, and we further assume that random vectors $\bs{x}$ are absolute continuous with probability density function (pdf) $p_\bs{x}(\bm{x})$ and expectation $\bb{E}\left[ \bs{x} \right] < \infty$. 
We use $\bold{0}$, $\bold{1}$, $\bm{e}_n$ and $\bm{I}$ to denote the vectors of all zeros, all ones, $n$-th Euclidean basis vector, and the identity matrix, where the size will be clear from the context. The $i$-th column, resp. $j$-th row, resp. submatrix of $i_1$-th to $i_2$-th column and $j_1$~-~th to $j_2$~-~th row, of a matrix is designated by $\bm{A}_{:,i}$, $\bm{A}_{j,:}$ and $\bm{A}_{i_1:i_2,j_1:j_2}$. $\rs{tr}\{\cdot\}$, $\odot$, $\oslash$ and $\ds{1}_{\mc{X}}: \bm{x}\to \{0,1\}$ denote the trace of a matrix, entrywise product, entrywise division and the indicator function defined as $\ds{1}_{\mc{X}}(\bm{x})=1$ if $\bm{x} \in \mc{X}$ and $0$ otherwise. $\ds{1}_+(x)$ is the Heaviside function and $(x)_+:=\rs{max}(0,x)$ is the rectified linear function.

\section{Optimal reconstruction by centroid computation of intersection polytopes}
\begin{figure}
\centering \subfigure[Intersection polytope $\mc{P}_\bm{t}$]{
  \includegraphics[width=0.9\linewidth]{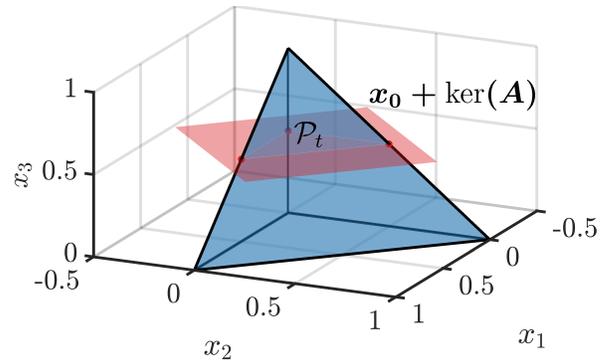}
} \hspace{0cm} \subfigure[Centroid $\bh{x}$]{
  \includegraphics[width=0.9\linewidth]{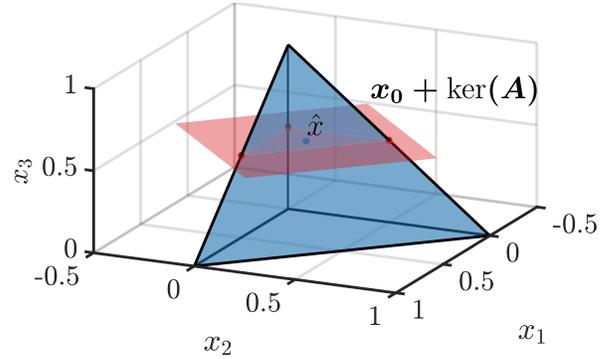}
}
\caption{Intersection polytope and centroid for $\bm{t} = 0.5$ and $\bm{A}~=\bm{V}_s^T=~[0,0,1]$.}
\label{fig:poly1}
\end{figure}

We assume that the sought vector $\bm{x}\in\bb{R}_+^N$ is a realization of a non-negative random vector $\bs{x}$ drawn from a joint distribution with a pdf denoted by $p_\bs{x}(\bm{x})$. Throughout the paper, we have the following assumption:
\begin{assumption}
\label{as:simplex_support}
The standard simplex $\Delta$ is the support of $p_\bs{x}(\bm{x})$ implying that
$p_\bs{x}(\Delta) = p_\bs{x}(\bm{x} \in \Delta) = 1$. 
\end{assumption}
This assumption imposes some compressibility on $\bm{x}$, which is often referred to as soft sparsity. In practice, input signals $\bm{x} \in \Delta$ appear freqeuently in applications that involve discrete probability vectors including softmax classification in machine learning \cite{GooBenCou16} as illustrated in Fig. \ref{fig:mnist_examples} and \ref{fig:cifar_examples}, multiple hypothesis testing in statistics \cite{GooBenCou16} or maximum likelihood decoding in communications \cite{TseVis05}. The problem of compressing and recovering such vectors applies in particular to distributed decision-making problems, where local decisions are distributed among different decision makers and need to be fused to improve an overall decision metric. We note that the set of discrete probability vectors of length $N+1$ generates the dihedral face of a simplex in dimension $N+1$ that can be embedded naturally into the (solid) simplex of dimension $N$ by removing one arbitrary component. We refer the reader to Sec. \ref{sec:numresults} for a more detailed description as well as numerical results of the proposed recovery method on soft-classification vectors generated by two standard datasets.

By Assumption \ref{as:simplex_support}, given a vector $\bm{y}=\bm{A} \bm{x}$, $\bm{A} \in \bb{R}^{M \times N}$, the set of feasible solutions is restricted to a polytope
\begin{align}\label{equ:poly1}
\mc{P}_\bm{y} := \Delta \cap \{ \bm{x} : \bm{A}\bm{x} = \bm{y} \} = \Delta \cap \left( \bm{x}_0 + \rs{ker}(\bm{A}) \right),\text{\footnotemark}
\end{align}
\footnotetext{For two sets $\mc{X}$ and $\mc{Y}$ we denote their Minkowski Sum by $\mc{X}+\mc{Y}:=\{\bm{x}+\bm{y}: \ \bm{x} \in \mc{X}, \bm{y} \in \mc{Y}\}$. Here, one set is the singleton $\bm{x}_0$.} where $\bm{x}_0$ is an arbitrary solution to $\bm{y} = \bm{A} \bm{x}$ (e.g. the Moore-Penrose pseudo inverse $\bm{x}_0= \bm{A}^{\dagger} \bm{y}$) and $\rs{ker}(\bm{A})$ denotes the $N-M$-dimensional kernel of $\bm{A}$. In other words, upon observing $\bm{y}\in\bb{R}^M$, the support of the conditional pdf $p_{\bs{x} \vert \bs{y}}(\bm{x}\vert\bm{y})$ is limited to $\mc{P}_\bm{y}$.
To simplify the subsequent derivations, let $\bm{V}_0$ and $\bm{V}_s$ constitute a basis of $\rs{ker}(\bm{A})$ and $\rs{ker}^{\perp}(\bm{A})$ as obtained, for instance, by the singular-value decomposition (SVD)
\begin{align}\label{equ:svd}
\bm{A} = \bm{U} \bm{\Sigma} \bm{V}^T = \begin{bmatrix} \bm{U}_s  \end{bmatrix} \begin{bmatrix} \bm{\Sigma}_s & \bold{0}  \end{bmatrix} \begin{bmatrix} \bm{V}_s^T \\ \bm{V}_0^T \end{bmatrix}.
\end{align}
Here, $\bm{U}_s\in\bb{R}^{M \times M}$, $\bm{\Sigma}_s \in \bb{R}^{M \times M}$, $\bm{V}_s \in \bb{R}^{N \times M}$ and $\bm{V}_0 \in \bb{R}^{N \times N-M}$.
Then, we apply \eqref{equ:svd} to \eqref{equ:poly1} and multiply $\bm{A}\bm{x} = \bm{y}$ from left by $\Sigma_s^{-1} \bm{U}_s^T$ (note that $\bm{U}_s^T \bm{U}_s = \bm{I}_M$) to obtain an equivalent description for the same polytope given by 
\begin{align}\label{equ:poly2}
\mc{P}_\bm{t} := \Delta \cap \{ \bm{x} : \bm{V}_s^T \bm{x} = \bm{t}\}, \ \rs{dim}(\mc{P}_\bm{t}) = N-M.
\end{align}
This description uses the equivalent measurement vector 
\begin{align}\label{equ:equivalent_measurements}
\bm{t} := \bm{\Sigma}_s^{-1} \bm{U}_s^T \bm{y} = \bm{V}_s^T \bm{x}
\end{align}
and defines the intersection polytope in terms of $\rs{ker}^{\perp}(\bm{A})$ via the orthogonal basis $\bm{V}_s$. We point out that as $\bm{A}$ is assumed to be real, so are also $\bm{U}$ and $\bm{V}$. Fig. \ref{fig:poly1}(a) depicts an example of the polytope $\mc{P}_\bm{t}$ for $M=1$ and $N=3$.

\begin{lemma}[Optimality of centroid estimator under uniform distribution]\label{lem:centest}
Assume $\bs{x} \sim \mc{U}(\Delta)$ and suppose that a compressed realization $\bm{t} = \bm{V}_s^T \bm{x}$ \eqref{equ:equivalent_measurements} has been observed. Let $\varrho_\mc{P}$ be the Lebesgue measure on the $(N-M)$-dimensional affine subspace that contains $\mc{P}_\bm{t}$ (see \cite[Sec. 2.4]{MakPod13}) and assume that $\rs{vol}(\mc{P}_\bm{t}):=\int_{\mc{P}_\bm{t}} 1 \ d\varrho_\mc{P} > 0$. Then, the conditional mean estimator $\bh{x}$ of $\bm{x}$ given $\bm{t}$ has the MMSE property
\begin{align}
\bb{E}\left[ \lVert \bs{x} - \bh{x} \rVert_2^2 \right] = \inf_f \bb{E}\left[ \lVert \bs{x} - f(\bm{t}) \rVert_2^2 \right]
\end{align}
and is obtained by
\begin{align}\label{equ:centest}
\hat{x}_n = \rs{vol}(\mc{P}_\bm{t})^{-1} \int_{\mc{P}_\bm{t}} {x}_n  \ d\varrho_\mc{P}, \ n\in\{1,\hdots,N\},
\end{align}
i.e., the centroid of the polytope $\mc{P}_\bm{t}$, cf. \eqref{equ:poly2}.
\end{lemma}
\begin{proof}
The proof is a standard result in estimation theory \cite{Kay93}.
\end{proof}

The remaining part of the paper is devoted to the following problem:
\begin{problem}\label{prob:centest}
Given a fixed measurement matrix $\bm{A}$, design a neural network composed of only elementary arithmetic operations and activation functions such that if the input to the network is $\bm{t}$, then its output is

\begin{enumerate}[(P1)]
\item \textit{the volume} $\rs{vol}(\mc{P}_\bm{t})$, and
\item \textit{the moments} $\bg{\mu} = \int_{\mc{P}_\bm{t}} \bm{x}  \ d\varrho_\mc{P}$. 
\end{enumerate}
\end{problem}
The sought neural network is depicted in Fig. \ref{fig:neural_computation_network}
\begin{figure}
\centering 
  \includegraphics[width=0.7\linewidth]{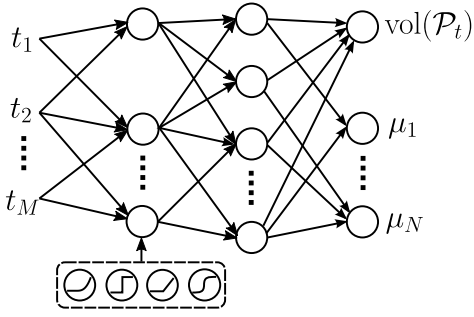}
\caption{Neural network to solve problem (P1) and (P2).}
\label{fig:neural_computation_network}
\end{figure}

\section{A review of Lasserre's Laplace techniques}
To make this paper as self-contained as possible and highlight the theoretical contribution of the original works \cite{LasZer01b,Las15b}, we start this section with a review of Laplace techniques that are used in Sec. \ref{sec:volnet} and \ref{sec:centnet} for computing the volume and moments of full- and lower-dimensional polytopes. First, we introduce a set of definitions restated from \cite{BryTuaGlaPru92,GlaPruSko06}.

\begin{definition}[Def. 1.4.3 \cite{GlaPruSko06}]\label{def:lt1}
The (one-sided) Laplace transform (LT) of a function $f: \bb{R}_+\to \bb{R}$ is the function $F:\bb{C}\to \bb{R}$ defined by
\begin{align}\label{equ:lt}
F(\lambda) = \mc{L}_t( f(t) ) := \int_0^\infty f(t) \exp(-\lambda t) \ dt
\end{align}
provided that the integral exists. 
\end{definition}
Here and hereafter, we refer to $t$ as transform variable and $\lambda$ as Laplace variable and conveniently designate LT transform pairs by $f(t) \laplace F(\lambda)$.

\begin{remark}\label{rem:ltpos}
We assume input functions can be written in terms of the Heaviside function as $\ds{1}_+(t) f(t)$ where $\ds{1}_+$ will usually be omitted. In other words, here and hereafter we will only consider functions that are supported on (a subset of) the non-negative real line as well as multidimensional generalizations supported on (subsets of) the non-negative orthant, respectively. We refer the reader to Sec. \ref{subsec:arbitrary_support} for a discussion of possible extensions to a larger class of functions with arbitrary support.
\end{remark}

\begin{lemma}[Def. 1.4.4, Th. 1.4.8 \cite{GlaPruSko06}]\label{lem:ltexistence}
Let $f$ be locally integrable and assume there exists a number $\gamma \in \bb{R}$ such that
\begin{align}\label{equ:lapconv}
\int_0^\infty \lvert f(t) \rvert \exp(- \gamma t ) \ d t < \infty.
\end{align}
Then, the Laplace integral \eqref{equ:lt} is absolutely and uniformly convergent on $\{ \lambda \in \bb{C} : \rs{Re(\lambda)} \geq \gamma \}$, and 
\begin{align}\label{equ:ltabscissa}
\gamma_\rs{ac} := \rs{inf} \biggl\{ \gamma \in \bb{R} : \int_0^\infty \lvert f(t) \rvert \exp(-\gamma t) \ dt < \infty \biggr\}
\end{align}
is called the abscissa of absolute convergence.
\end{lemma}

\begin{lemma}[Th. 1.4.12]\cite{GlaPruSko06}\label{lem:iltexistence}
Let $f$ be as in Lemma \ref{lem:ltexistence} and (possibly piecewise) continuous. Then in points $t$ of continuity the inverse Laplace transform (ILT) is given by
\begin{align}\label{equ:iltcontinuous}
f(t) = \mc{L}_\lambda^{-1}( F(\lambda) ) := \frac{1}{2 \pi i} \int_{c -i \infty}^{c + i \infty} F(\lambda) \exp(\lambda t) \ d\lambda, \ c>\gamma_\rs{ac}
\end{align}
and in points of discontinuity we have
\begin{align}
\mc{L}_\lambda^{-1}( F(\lambda) ) & = \frac{f(t^+) + f(t^-)}{2}.
\end{align}
\end{lemma}
As $f(t)= 0$ for $t<0$, $f(t)$ in $f(t) \laplace F(\lambda)$ is to be understood as $\ds{1}_+(t) f(t)$.

\begin{definition}\label{def:mdim_lt}
The (one-sided) $M$-dimensional LT ($M$D-LT) of a function $f:\bb{R}_+^M \to \bb{R}$ is the function $F:\bb{C}^M \to \bb{C}$ defined by a concatenation of LTs
\begin{align}\label{equ:mdim_lt}
F(\bg{\lambda}) = \left( \circ_{m=1}^M \mc{L}_{t_m} \right)\left( f(\bm{t}) \right)
\end{align}
provided that all integrals exist. If the individual LTs converge absolutely and uniformly, the order of integration in \eqref{equ:mdim_lt} is arbitrary \cite{BryTuaGlaPru92}.
\end{definition}

\begin{definition}\label{def:mdim_ilt}
The $M$-dimensional ILT ($M$D-ILT) of a function $F(\bg{\lambda}): \bb{C}^M \to \bb{C}$ is defined by the concatenation of ILTs
\begin{align}\label{equ:mdim_ilt}
f(\bm{t}) = \left( \circ_{m=1}^M \mc{L}_{\lambda_m}^{-1} \right)\left( F(\bg{\lambda}) \right)
\end{align}
provided that all integrals exist. If the individual ILTs converge absolutely and uniformly the order of integration is again arbitrary. The function $f(\bm{t})$ in expressions like $f(\bm{t}) \laplace F(\bg{\lambda})$ is to be understood as $(\prod_{m=1}^M \ds{1}_+(t_m)) f(\bm{t})$.
\end{definition}

\begin{remark}\label{rem:ommit_existence_conditions}
To make this paper less technical, and therefore more accessible for a broader audience, we omit a more detailed exposition of operational properties as well as conditions on the existence of the (multidimensional) Laplace operators. We refer the interested reader to \cite{BryTuaGlaPru92,GlaPruSko06}, and point out that the transforms appearing in this article are obtained by combining standard transform pairs summarized in Tab. \ref{tab:tab_lts}.
\end{remark}

Now we are in a position to introduce the general idea of Lasserre for polyhedral volume computation \cite{LasZer01b,Las15b} that consists in exploiting the identity
\begin{align}\label{equ:lapid}
f(\bm{t}) & = \left( \circ_{m} \mc{L}_{\lambda_m}^{-1} \right)\left( \circ_{m} \mc{L}_{t_m} \right)(f(\bm{t})) \nonumber \\
 & = \bigr( \underbrace{\mc{L}_{\lambda_M}^{-1} \circ \hdots \circ \mc{L}_{\lambda_1}^{-1}}_{M-\rs{times}} \circ \underbrace{\mc{L}_{t_M} \circ \hdots \circ \mc{L}_{t_1}}_{M-\rs{times}} \bigl) ( f(\bm{t}) )
\end{align}
for functions $f(\bm{t})$ admitting (multidimensional) Laplace transforms according to Lemma \ref{lem:ltexistence}, \ref{lem:iltexistence} and Def. \ref{def:mdim_lt}, \ref{def:mdim_ilt}.

Interestingly, this approach allows for evaluating complicated functions as $\rs{vol}(\mc{P}_\bm{t})$ without resorting to costly multidimensional numerical integration. For the particular case of (P1) and a single measurement (see illustration in Fig. \ref{fig:poly1}) with 
\begin{align}
f(t)= \rs{vol}( \Delta \cap \{ \bm{x} : \bm{a}^T \bm{x} = t\})
\end{align}
we can apply the following result of Lasserre \cite{Las15b}:

\begin{lemma}[Volume of a simplex slice \cite{Las15b}]\label{lem:volsimp}
Let $\mc{S} = \{ \bm{s}_1, \hdots \bm{s}_{N+1}\} = \{\bm{e}_1, \hdots, \bm{e}_N, \bold{0}\}$ denote the vertices of the standard simplex and $\bm{a}\in \bb{S}^{N-1}$ be such that $\langle \bm{a}, \bm{s}_n \rangle \neq \langle \bm{a}, \bm{s}_{n'} \rangle$ for any pair of distinct vertices $\bm{s}_n$, $\bm{s}_{n'}$. Then, the volume of the simplex slice at any point $t \in \bb{R}$ is given by
\begin{align}\label{equ:1dpoly}
\rs{vol}\left( \mc{P}_t \right) 
= \frac{1}{(N-1)!} \sum_{n=1}^{N+1}  \frac{ (t - \langle \bm{a}, \bm{s}_n \rangle)_+^{N-1}}{ \prod_{n\prime \neq n} ( \langle \bm{a}, \bm{s}_{n\prime} \rangle - \langle \bm{a}, \bm{s}_n \rangle )}.
\end{align}
\end{lemma}
\begin{proof}
The lemma is a restatement of Thm. 2.2 in \cite{Las15b}.
\end{proof}

\begin{figure}
\centering 
  \includegraphics[width=0.55\linewidth]{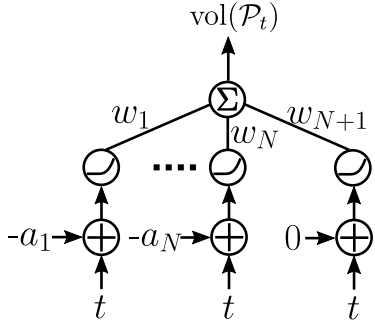}
\caption{Volume computation network with \textit{rectified polynomial} layers for input $t\in\bb{R}$.}
\label{fig:computation_tree_onelayer}
\end{figure}
On closer inspection of \eqref{equ:1dpoly}, we find that $\rs{vol}\left( \mc{P}_t \right)$ can be efficiently implemented given some fixed $\bm{a}$ and an input $t\in \bb{R}$ by a $1$-layer (shallow) neural network via a set of \textit{rectified polynomial} activation functions with shift $\{-a_1,\hdots,-a_N,0\}$ and weights $\bm{w}_n = \prod_{n\prime \neq n} ( \langle \bm{a}, \bm{s}_{n\prime} \rangle - \langle \bm{a}, \bm{s}_n \rangle )^{-1}$ as illustrated in Fig. \ref{fig:computation_tree_onelayer}. 
\begin{remark}\label{rem:tfcaffe_implem}
Even though the \textit{rectified polynomial} activation function $(t - a_n)_+^{N-1}$ is currently not implemented in common Deep Learning libraries, it can be efficiently computed via a concatenation of a conventional ReLU $(t-a_n)_+$ followed by a polynomial activation function $(t-a_n)^{N-2}$ (both supported in e.g. the Caffe or Tensorflow framework \cite{JiaSheDonKar14,Aba15}).
\end{remark}

\subsection{Extending Laplace techniques for arbitrary supports of $f(t)$}\label{subsec:arbitrary_support}
It was shown in \cite{Las15b} that \eqref{equ:1dpoly} also holds when some $a_n<0$ with $\exists t<0$ such that $f(t)>0$ preventing the direct application of Laplace transform techniques (see Rem. \ref{rem:ltpos}). The necessary extension was obtained in \cite{Las15b} essentially by introducing a translation operator $S_{{t}^*}(\cdot)$ defined by
\begin{align}\label{equ:translation_operator}
S_{{t}^*}(f({t})) := f({t}-{t}^\star)
\end{align}
and a translation identity in conjunction with \eqref{equ:lapid} given by
\begin{align}\label{equ:laplace_translation_identity}
f(t) = \left( S_{{t}^*} S_{{-t}^*} \right) ( f(t) ) = \left( S_{{t}^*} \mc{L}_\lambda^{-1} \mc{L}_t S_{{-t}^*} \right) ( f(t) ).
\end{align}
Here, the shift $t^\star \in \bb{R}$ has to be chosen such that the identity
\begin{align}\label{equ:shifted_laplace_identity}
S_{{-t}^*} (f(t)) =  \left( \mc{L}_\lambda^{-1} \mc{L}_t   S_{{-t}^*} \right) ( f(t) )
\end{align}
holds and the LT acts on a function $f$ vanishing for $t<0$, i.e., $f(t)=0$ for $t<0$.

\begin{remark}\label{rem:nonneg_Vs}
To avoid the obfuscation connected with a multidimensional extension of the shifted LT identity \eqref{equ:shifted_laplace_identity} and involved analysis of admissibility conditions we consider in the remainder of the paper only a particular case specified by the following assumption. We show by numerical simulations that similar to Lemma \ref{lem:volsimp} the neural networks to be introduced in the following indeed apply for every orthogonal basis $\bm{V}_s$.
\end{remark}
\begin{assumption}\label{ass:nonneg_Vs}
$\bm{A}$ admits a non-negative orthogonal basis $\bm{V}_s\in \bb{R}_{+}^{N \times M}$ for $\rs{ker}^{\perp}(\bm{A})$.
\end{assumption}

\begin{table}\normalsize{
\begin{center}
  \begin{tabular}{l l  l }
  	\# & $f(t)$ & $F(\lambda)$  \\[1ex] \hline
  	(LT1) &$c_1 f_1(t) + c_2 f_2(t)$ & $c_1 F_1(\lambda) + c_2 F_2(\lambda)$  \\[1ex]
  	(LT2) &$\dfrac{t^{n-1}}{(n-1)!}$ & $\dfrac{1}{\lambda^{n}}$ \\[1ex]
  	(LT3) &$\exp(at)$ & $\dfrac{1}{\lambda - a}$ \\[1ex]
  	(LT4) &$t\exp(at)$ & $\dfrac{1}{(\lambda - a)^{2}}$ \\[1ex]
  	(LT5) &$\dfrac{\exp(at)-\exp(bt)}{a-b}\ (a\neq b)$ & $\dfrac{1}{(\lambda - a)(\lambda - b)}$ \\[1ex]
  	(LT6) &$\ds{1}_+(t-a) f(t-a)\ (a\geq 0)$ & $\exp(-a\lambda) F(\lambda)$  \\[1ex]
  	(LT7) &$f'(t)$ & $\lambda F(\lambda) - f(0^+)$  
  \end{tabular}
  \end{center}
  \caption{Table of Laplace Transforms}
  \label{tab:tab_lts}
  }
  \end{table}
  
\section{Volume computation network}\label{sec:volnet}
\subsection{Theoretical Foundation}
The goal of this section is to solve (P1), i.e., to design a neural network that computes $\rs{vol}(\mc{P}_\bm{t})$ for some given $\bm{t} \in \bb{R}^M$. By Assumption \ref{ass:nonneg_Vs}, we have $\bm{V}_s \in \bb{R}_{+}^{N\times M}$, $\bm{V}_s^T \bm{V}_s = \bm{I}$. With this in hand, it may be verified that $\rs{vol}(\mc{P}_\bm{t}) = 0$ whenever there is some $t_m < 0$ and Laplace transformations including the identity \eqref{equ:lapid} can be applied. To compute $\rs{vol}(\mc{P}_\bm{t})$, we start with a lemma on exponential integrals over the simplex. This lemma will be used later on.

\begin{lemma}\label{lem:baldoni_int}
Let $\bm{l}\in\bb{R}^N$ be a linear form such that $\langle \bm{l}, \bm{s}_n \rangle \neq \langle \bm{l}, \bm{s}_{n'} \rangle$ for any pair of distinct vertices $\bm{s}_n$, $\bm{s}_{n'}\in\Delta$. Then we have
\begin{align}\label{equ:baldoni_int}
\int_\Delta \exp(-\langle \bm{l}, \bm{x} \rangle) \ d\bm{x} = \sum_{n=1}^{N+1} \frac{ \exp( - \langle \bm{l}, \bm{s}_n \rangle ) }{ \prod_{n'\neq n} \langle \bm{l}, \bm{s}_{n'} - \bm{s}_n \rangle}.
\end{align}
\end{lemma}
\begin{proof}
The lemma follows from \cite[Cor. 12]{BalBerDeKoe11} by changing the sign of the linear form and noting that the volume of the (full-dimensional) simplex is equal to $(N!)^{-1}$.
\end{proof}

We use this result to compute the inner $M$D-LT, cf. \eqref{equ:mdim_lt}, in \eqref{equ:lapid}. To this end, let $\mc{T}_\bm{t}:=~\Delta \cap \{\bm{x}:~\bm{V}_s^T\bm{x}\leq\bm{t}\}$ be the intersection of the simplex with $M$ halfspaces, and consider
\begin{align}
F(\bg{\lambda}) = F(\lambda_1,\hdots,\lambda_M) = \bigr( \underbrace{\mc{L}_{t_M} \circ \hdots \circ \mc{L}_{t_1}}_{M-\rs{times}} \bigl) \left( \rs{vol}(\mc{T}_\bm{t} ) \right).
\end{align}
By the definition of the $M$D-LT \eqref{equ:mdim_lt} and the fact that $\bb{R} \to \bb{R}_+: x_k \to \exp(x_k)$ is non-negative, the Fubini-Tonelli theorem \cite[Th. 9.11]{Tes14} implies that, for $\rs{Re}(\bg{\lambda}) > \bold{0}$, we have
\begin{align}\label{equ:lapvolTt}
F(\bg{\lambda}) & = \int_{\bb{R}_+^M} \exp(-\langle \bg{\lambda}, \bm{t} \rangle) \left( \int_{\Delta \cap \{\bm{x} : \bm{V}_s^T \bm{x} \leq \bm{t}\}} 1 \ d\bm{x} \right) \ d\bm{t} \nonumber \\
& = \int_{\Delta} \left( \int_{[\bm{v}_{1}^T \bm{x}, \infty) \times \hdots \times [\bm{v}_{M}^T \bm{x}, \infty)} \exp(-\langle \bg{\lambda}, \bm{t} \rangle) \ d\bm{t} \right) \ d\bm{x} \nonumber \\
& = \frac{1}{\prod_{m=1}^M \lambda_m} \int_{\Delta} \exp(-\langle \bm{V}_s \bg{\lambda}, \bm{x} \rangle) \ d\bm{x} .
\end{align}
So applying the Laplace identity \eqref{equ:lapid} with $f(\bm{t})=\rs{vol}(\mc{T}_\bm{t})$ to \eqref{equ:lapvolTt} shows that
\begin{align}\label{equ:volTt}
\rs{vol}(\mc{T}_\bm{t}) = \left( \circ_{m} \mc{L}_{\lambda_m}^{-1} \right) \left(\frac{1}{\prod_{m} \lambda_m} \int_{\Delta} \exp(-\langle \bm{V}_s \bg{\lambda}, \bm{x} \rangle) \ d\bm{x} \right),
\end{align}
whenever the $M$D-ILT \eqref{equ:mdim_ilt} on the RHS exists. In this case, the RHS is a function of the transform variable $\bm{t}$ and the inner simplex integral may be evaluated via Lemma \ref{lem:baldoni_int} if the corresponding condition holds. To obtain $\rs{vol}(\mc{P}_\bm{t})$, we use \eqref{equ:lapvolTt} to establish the following proposition.

\begin{proposition}\label{prop:volPt}
Let $\bm{V}_s$ be a non-negative orthogonal basis ($\bm{V}_s\in\bb{R}_+^{N\times M}$, $\bm{V}_s^T \bm{V}_s = \bm{I}_M$) and $\mc{P}_\bm{t}:=\Delta \cap \{\bm{x}: \bm{V}_s^T \bm{x} = \bm{t}\}$. Then, we have
\begin{align}\label{equ:lapvolPt}
\rs{vol}(\mc{P}_\bm{t}) = \left( \circ_{m=1}^M \mc{L}_{\lambda_m}^{-1} \right) \left( \int_{\Delta} \exp(-\langle \bm{V}_s \bg{\lambda}, \bm{x} \rangle) \ d\bm{x} \right)
\end{align}
provided that the integrals on the RHS exist.
\end{proposition}
\begin{proof}
The proof is deferred to Appendix \ref{app:appA}.
\end{proof}

Now let us turn our attention to the numerical evaluation of $\rs{vol}(\mc{P}_\bm{t})$ via the $M$D-ILT \eqref{equ:mdim_ilt}. To this end, we first evaluate the inner integral over the simplex using Lemma \ref{lem:baldoni_int} to obtain
\begin{align}\label{equ:volPt_baldoni}
\rs{vol}(\mc{P}_\bm{t}) = \left( \circ_{m=1}^M \mc{L}_{\lambda_m}^{-1} \right) \left( \sum_{n=1}^N \frac{ \exp( -\langle \bm{V}_s \bg{\lambda}, \bm{s}_n \rangle ) }{ \prod_{n'\neq n} \langle \bm{V}_s \bg{\lambda}, \bm{s}_{n'} - \bm{s}_n \rangle } \right).
\end{align}

Given that the argument of the $M$D-ILT in \eqref{equ:volPt_baldoni} is a sum of \textit{exponential-over-polynomial} (\textit{exp-over-poly}) functions we continue with a lemma on corresponding one-dimensional transform pairs.

\begin{lemma}[ILT of an \textit{exp-over-poly} function]\labeli[lem:exp_invlap]
\hfill 
\begin{enumerate}[1)]
\item\labelii[lem:exp_invlapM1] {Let $M=1$, $a\geq 0$, $\bm{b}\in \bb{R}_{\neq 0}^N$. Then we have the transform pair ($F(\lambda) \Laplace f(t)$)
\begin{align}\label{equ:ilt_exppoly_M1}
\frac{\exp( -a \lambda )}{ \prod_{n=1}^N (b_n \lambda) } \ \overset{ a \geq 0}{\Laplace} \ \frac{ (t - a)_+^{N-1} }{(N-1)! \prod_{n=1}^N b_n }.
\end{align}}

\item\labelii[lem:exp_invlapM2] {Let $M\geq 2$, $\bm{a} \in \bb{R}^{M}$ and $\bm{B}_{:,1} \in \bb{R}_{\neq 0}^N$, $\bm{B}_{:,2:M} \in \bb{R}^{N \times M-1}$ with pairwise linearly independent rows. Then we have the transform pair $F(\lambda_1,\hdots,\lambda_M) \Laplace f(t_1,\lambda_2,\hdots,\lambda_M)$

\begin{align}\label{equ:ilt_exppoly_M2}
\frac{\exp( -\langle \bm{a}, \bg{\lambda} \rangle )}{ \prod_{n=1}^N [\bm{B} \bg{\lambda}]_n }  \overset{ a_1 \geq 0}{\Laplace}  \frac{ \ds{1}_+(t_1 - a_1) }{\prod_{n} B_{n,1} } \sum_{n=1}^N \frac{ \exp( - \langle \bm{a}^{(n)} , \bg{\lambda}_{2:M} \rangle ) } {\prod_{n'=1}^{N-1} [\bm{B}^{(n)} \bg{\lambda}_{2:M}]_{n'}}.
\end{align}
Setting $\bm{C} := \bm{B} \oslash (\bm{B}_{:,1} \bold{1}^T)$ we obtain $\bm{a}^{(n)}\in \bb{R}^{M-1}$ and $\bm{B}^{(n)} \in \bb{R}^{N-1 \times M-1}$ ($n\in\{1,\hdots,N\})$ by
\begin{align}
\bm{a}^{(n)} & = \bm{a}_{2:M} + (t_1 - a_1)\bm{C}_{n,2:M}, \\
\bm{B}^{(n)} & = \bm{C}_{1:N\setminus n,2:M}  - \bold{1} \bm{C}_{n,2:M}.
\end{align}
}

\end{enumerate}
\end{lemma}

\begin{remark}\label{rem:trunc_concat}
We highlight that in case ${B}_{k,1} = 0$ for some $k$ we can treat the corresponding factor $\bm{B}_{k,:} \bg{\lambda} = \bm{B}_{k,2:M} \bg{\lambda}_{2:M}$ as a constant w.r.t. the ILT $\mc{L}_{\lambda_1}^{-1}(\cdot )$. Accordingly, we can apply \eqref{equ:ilt_exppoly_M2} using the truncated matrix $\bt{B} = \bm{B}_{1:N \setminus k,:}$, where $\bt{B}_{:,1} \in \bb{R}_{\neq 0}^{N-1}$, and obtain the truncated output $\bt{B}^{(n)} \in \bb{R}^{N-1 \times M-1}$. For the subsequent ILT we have to include the corresponding factor $\bm{B}_{k,2:M} \bg{\lambda}$ again by applying the matrix concatenation
\begin{align}
\bm{B}^{(n)} = \begin{bmatrix} \bt{B}^{(n)} \\ \bm{B}_{k,2:M} \end{bmatrix}.
\end{align}
\end{remark}

\subsection{Structure of the network}
\begin{figure}
\centering 
  \includegraphics[width=\linewidth]{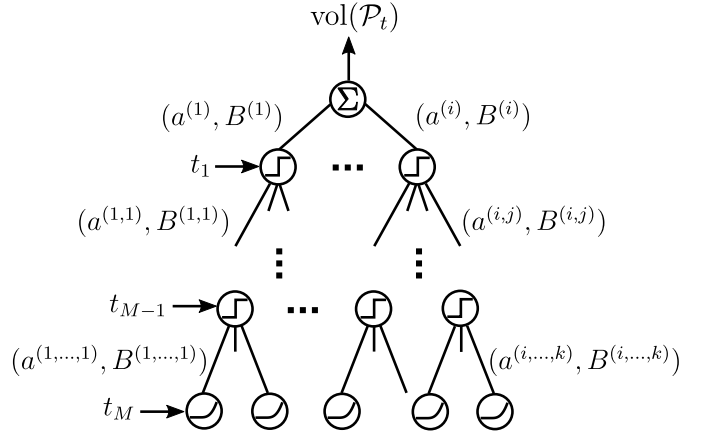}
\caption{Volume network with \textit{rectified polynomial} and \textit{threshold} layers for input $\bm{t}\in \bb{R}^M$.}
\label{fig:computation_tree}
\end{figure}
Now we are in a position to present a computation network for the $M$D-ILT \eqref{equ:lapvolPt} by an algorithm that can be described by a  computational tree (see \cite{LasZer01b}). Due to the particular form of RHS in \eqref{equ:baldoni_int} the root of this tree is the $M$D-ILT of $N+1$ \textit{exp-over-poly}-functions
\begin{align}
(\bm{a}^{(i)},\bm{B}^{(i)}) \to \frac{ \exp(-\langle \bm{a}^{(i)}, \bg{\lambda} \rangle) }{ \prod_{n'} [\bm{B}^{(i)} \bg{\lambda}]_{n'} }.
\end{align}
\begin{enumerate}
\item In layer $m=1$, each node computes the ILT $F(\lambda_1,\hdots,\lambda_M) \Laplace f(t_1,\lambda_2,\hdots,\lambda_M)$ of an \textit{exp-over-poly}-function with parameters $(\bm{a}^{(i)},\bm{B}^{(i)})$ inherited from the root node. If the branch is active, i.e. $\ds{1}_+(t_1 - a_1^{(i)})(\prod_n B_{n,1}^{(i)})^{-1} \neq 0$, it applies transform \eqref{equ:ilt_exppoly_M2} and passes the corresponding $N$ exp-over-poly functions $( \bm{a}^{(i,j)},\bm{B}^{(i,j)})$, $(i,j)\in \{1,\hdots,N+1\}\times\{1,\hdots,N\}$ to its children nodes.
\item In layer $m=2$, each node computes the ILT $f(t_1,\lambda_2,\hdots,\lambda_M) \Laplace f(t_1,t_2,\lambda_3,\hdots,\lambda_M)$ of an \textit{exp-over-poly}-function inherited from its parent node.
If the branch is active, i.e. $\ds{1}_+(t_2 - a_1^{(i,j)})(\prod_n B_{n,1}^{(i,j)})^{-1} \neq 0$, it applies transform \eqref{equ:ilt_exppoly_M2} and passes the corresponding $N-1$ exp-over-poly functions $( \bm{a}^{(i,j,k)},\bm{B}^{(i,j,k)})$, $(i,j,k)\in \{1,\hdots,N+1\}\times\{1,\hdots,N\}\times\{1,\hdots,N-1\}$ to its children nodes.
\item ...
\item In the last layer $m=M$, each node computes the ILT $f(t_1,\hdots,t_{M-1},\lambda_M) \Laplace f(t_1,\hdots,t_M)$ of an \textit{exp-over-poly}-function inherited from its parent node. If the branch is active, i.e. $\ds{1}_+(t_M - a_1^{(i,j,\hdots)}) ( (N-1)!(\prod_n B_{n,1}^{(i,j,\hdots)}) )^{-1} \neq 0$, it applies transform \eqref{equ:ilt_exppoly_M1} and obtains a numerical value.
\item The final result is obtained by summing and weighting all values of the last layer.
\end{enumerate}
The described computational tree is equivalent to a deep neural network with $M$ layers, weights $\prod_n {B}_{n,1}^{(i,j,\hdots)}$ 
and \textit{threshold} as well as \textit{rectified polynomial} activation functions (see Fig. \ref{fig:computation_tree}).

\subsection{Volume computation example}
In the following we present a numerical example for the computation of $\rs{vol}(\mc{P}_\bm{t})$ via the described computation network and repeated application of Lemma \ref{lem:exp_invlap}. 
\begin{figure}
\centering 
  \includegraphics[width=0.9\linewidth]{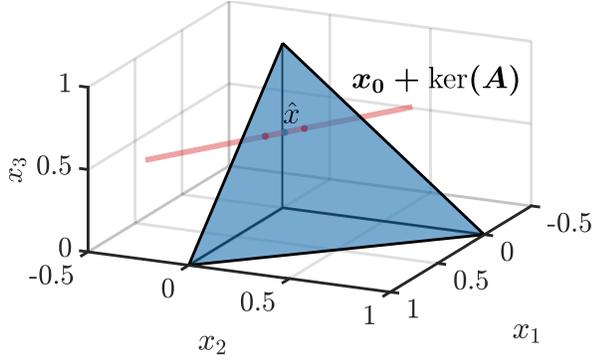}
\caption{Intersection polytope (line) and centroid for $\bm{t} = [0.5;0.933]$ and $\bm{A}=\bm{V}_s^T= [0,0,1;0.5,0.866,0]$.}
\label{fig:poly3}
\end{figure}

\begin{example}[Volume computation]
Assume $M=2$, $N=3$ and let $\mc{P}_\bm{t}$ be defined by the measurements (see corresponding illustration in Fig. \ref{fig:poly3})
\begin{align}
\bm{t} &= \bm{V}_s^T \bm{x} \nonumber \\
\begin{bmatrix} 0.5 \\ 0.0933 \end{bmatrix} &= \begin{bmatrix} 0 & 0 & 1 \\ 0.5 & 0.866 & 0 \end{bmatrix} \begin{bmatrix} 0.1 \\ 0.05 \\ 0.5 \end{bmatrix}.
\end{align}
The exp-over-poly functions at the root node are obtained by \eqref{equ:baldoni_int} and defined by the following parameters:
\begin{align*}
\bm{a}^{(1)} & = \begin{bmatrix} 0 \\ 0.5 \end{bmatrix}, \ \bm{a}^{(2)} = \begin{bmatrix} 0 \\ 0.866 \end{bmatrix}, \
\bm{a}^{(3)} = \begin{bmatrix} 1 \\ 0 \end{bmatrix} \ \bm{a}^{(4)} = \begin{bmatrix} 0 \\ 0 \end{bmatrix} \\
\bm{B}^{(1)} & = \begin{bmatrix} 0 & 0.366 \\ 1 & -0.5 \\ 0 & -0.5 \end{bmatrix}, \ \bm{B}^{(2)} = \begin{bmatrix} 0 & -0.366 \\ 1 & -0.866 \\ 0 & -0.866 \end{bmatrix}, \\
\bm{B}^{(3)} & = \begin{bmatrix} -1 & 0.5 \\ -1 & 0.866 \\ -1 & 0 \end{bmatrix}, \ \bm{B}^{(4)} = \begin{bmatrix} 0 & 0.5 \\ 0 & 0.866 \\ 1 & 0 \end{bmatrix}.
\end{align*}
The next step is to compute the ILT with respect to $\lambda_1$, i.e.,
\begin{align}\label{equ:example_ilt}
\mc{L}_{\lambda_1}^{-1} \left( \sum_{i=1}^{4} \frac{ \exp( -\langle \bm{a}^{(i)}, \bg{\lambda} \rangle) }{ \prod_{n=1}^{N} [\bm{B}^{(i)} \bg{\lambda} ]_{n} } \right). 
\end{align}
Applying \eqref{equ:ilt_exppoly_M2} to \eqref{equ:example_ilt} and using truncation and concatenation when required (see Rem. \ref{rem:trunc_concat}) yields
\begin{align*}
\bm{a}^{(1,1)} &= 0.25, \ 
\bm{a}^{(2,1)} = 0.433, \
\bm{a}^{(3,1)} = 0.25, \\
\bm{a}^{(3,2)} &= 0.433, \
\bm{a}^{(3,3)} = 0, \
\bm{a}^{(4,1)} = 0, \\
\bm{B}^{(1,1)} &= \begin{bmatrix} 0.366 \\ -0.5 \end{bmatrix}, \
\bm{B}^{(2,1)} = \begin{bmatrix} -0.366 \\ -0.866 \end{bmatrix}, \
\bm{B}^{(3,1)} = \begin{bmatrix} -0.366 \\ 0.5 \end{bmatrix}, \\
\bm{B}^{(3,2)} &= \begin{bmatrix} 0.366 \\ 0.866 \end{bmatrix}, \
\bm{B}^{(3,3)} = \begin{bmatrix} -0.5 \\ -0.866 \end{bmatrix}, \
\bm{B}^{(4,1)} = \begin{bmatrix} 0.5 \\ 0.866 \end{bmatrix}.
\end{align*}
Accordingly, the outputs of the last layer in the computational tree are given by
\begin{align*}
f^{(i,j)} & = \frac{ (t_2 - a_1^{(i,j)})_+}{\prod_n B_{n,1}^{(i,j)}} \cdot \frac{ \ds{1}_+(t_1 - a_1^{(i)}) }{\prod_{n'} B_{n',1}^{(i)}} \\
f^{(1,1)} & = 0, \ f^{(2,1)} = 0,\ f^{(3,1)} = 0 \\
f^{(3,2)} & = 0, \ f^{(3,3)} = 0,\ f^{(4,1)} = 0.2155,
\end{align*}
where index $n' \in \rs{supp}(\bm{B}_{:,1}^{(i)})$ and the final result is given by $\rs{vol}(\mc{P}_\bm{t})=0.2155$.
\end{example}

\section{Moment computation network}\label{sec:centnet}
The final step towards the optimal estimator of Lemma \ref{lem:centest} is to compute the moment vector $\bg{\mu} = \int_{\mc{P}_\bm{t}} \bm{x} \ d\varrho_\mc{P}$. To this end, we introduce an extension of Lemma \ref{lem:baldoni_int} as well as a conjecture that was verified numerically but currently lacks a rigorous proof.

\begin{lemma}\label{lem:baldoni_extended}
Let $\bm{l}$ and $\Delta$ be as in Lemma \ref{lem:baldoni_int}. Then we have
\begin{align}\label{equ:baldoni_extended}
& \int_{\Delta} x_k \exp(-\langle \bm{l}, \bm{x} \rangle) \ d\bm{x}  \nonumber \\
& = \frac{ \exp( - \langle \bm{l}, \bm{s}_k \rangle )}{ \prod_{n \neq k} \langle \bm{l}, \bm{s}_{n} - \bm{s}_k \rangle } + \sum_{\substack{n=1 \\ n\neq k}}^{N+1} \frac{ \exp( -\langle \bm{l}, \bm{s}_{n}  \rangle ) }{ \langle \bm{l}, \bm{s}_k - \bm{s}_n \rangle \prod_{n' \neq n} \langle \bm{l}, \bm{s}_{n'} - \bm{s}_n \rangle } \nonumber \\
& - \sum_{\substack{n = 1 \\ n \neq k}}^{N+1} \frac{ \exp( -\langle \bm{l}, \bm{s}_{k}  \rangle ) }{ \langle \bm{l}, \bm{s}_n - \bm{s}_k \rangle \prod_{n' \neq k} \langle \bm{l}, \bm{s}_{n'} - \bm{s}_k \rangle }. \nonumber \\
\end{align}
\end{lemma}
\begin{proof}
The proof is deferred to Appendix \ref{app:appC}.
\end{proof}

\begin{conjecture}\label{conj:centPt}
Let $\mc{P}_{\bm{t}}$ and $\bm{V}_s$ be as in Prop. \ref{prop:volPt}. Then, the moment $\mu_k:=\int_{\mc{P}_\bm{t}} x_k \ d\varrho_\mc{P} = \rs{vol}(\mc{P}_\bm{t}) \hat{x}_k$ (see \eqref{equ:centest}) is given by
\begin{align}\label{equ:lapmomPt}
\mu_k = \left( \circ_{m} \mc{L}_{\lambda_m}^{-1} \right) \left( \int_{\Delta} x_k \exp(-\langle \bm{V}_s \bg{\lambda}, \bm{x} \rangle) \ d\bm{x} \right),
\end{align}
provided that the integrals on the RHS exist.
\end{conjecture}
\begin{remark}
Conj. \ref{conj:centPt} originates from Prop. \ref{prop:volPt} and its numerical correctness was verified by extensive comparison with integration by simplicial decomposition \cite{BueEngFuk00}. In order to close a gap in a rigorous proof, we need to relate the moment of a polyhedron to the moment of an $(N-M)$-dimensional face in the spirit of \eqref{equ:polyface_differentiate} \cite[Prop. 3.3]{Las83}. 
\end{remark}
\begin{assumption}
Conj. \ref{conj:centPt} is assumed to be valid in what follows.
\end{assumption}

\begin{figure}
\centering 
  \includegraphics[width=\linewidth]{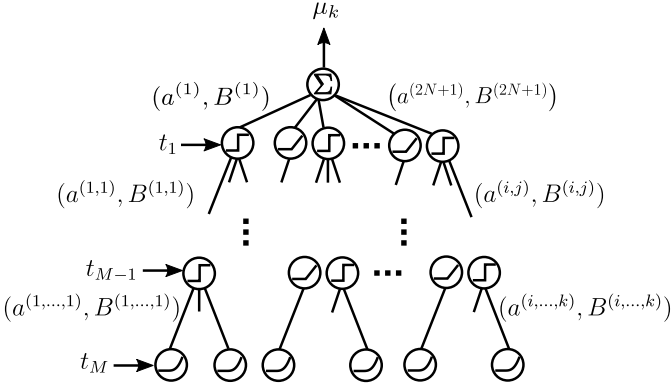}
\caption{Moment network for $\bg{\mu}_k$ with \textit{rectified linear}, \textit{rectified polynomial} and \textit{threshold} layers for $\bm{t}\in \bb{R}^M$.}
\label{fig:computation_tree_moment}
\end{figure}

To evaluate the $M$D-ILT in \eqref{equ:lapmomPt} we inspect the terms in the sum \eqref{equ:baldoni_extended} and note that the first term contains only simple poles permitting its Laplace transformation via Lemma \ref{lem:exp_invlap}. In fact, the corresponding term already appears in \eqref{equ:baldoni_int} and accordingly the $M$D-ILT is already computed as part of the volume computation network.
However, this does not apply to the remaining $2N$ terms as they contain quadratic terms $\langle \bm{l}, \bm{s}_k - \bm{s}_n \rangle$ and $\langle \bm{l}, \bm{s}_n - \bm{s}_k \rangle$. Accordingly, for $M\geq 2$ two rows of the denominator matrix $\bm{B}^{(i)}$, $2\leq i \leq 2N +1$ will be linearly dependent violating the assumptions of Lemma \ref{lem:exp_invlap}. The following proposition provides a modified result applicable to ILTs of \textit{exp-over-poly}-functions with one double pole (resp. two linearly dependent rows of $\bm{B}$).\footnote{While higher order (e.g. cubic, quartic) terms may appear when employing particular structured measurement matrices like partial Hadamard or discrete cosine transform matrices, these cases may be reduced to the considered quadratic case when a small numeric perturbation is applied.}

\begin{proposition}[ILT of \textit{exp-over-poly} function with one double pole]\label{prop:exp_invlap_dpole}
Let $M\geq 2$, $N\geq 3$, $\bm{B}_{:,1}\in \bb{R}_{\neq {0}}^N$ and assume w.l.o.g. that the first and second row of $\bm{B}$ are equal, i.e., $\bm{B}_{1,:}=\bm{B}_{2,:}$. Then, the transform pair $F(\lambda_1,\hdots,\lambda_M) \Laplace f(t_1,\lambda_2,\hdots,\lambda_M)$ is given by
\begin{align}\label{equ:exp_invlap_M2_dpole}
\frac{\exp( -\langle \bm{a}, \bg{\lambda} \rangle )}{ \prod_{n=1}^N [\bm{B} \bg{\lambda}]_n }  \overset{ a_1 \geq 0}{\Laplace} 
\frac{ (t_1 - a_1)_+ }{\prod_{n=1}^N B_{n,1}} \frac{ \exp( -\langle \bm{a}^{(1)}, \bg{\lambda}_{2:M} \rangle) }{ \prod_{n'=1}^{N-1} [\bm{B}^{(1)} \bg{\lambda}_{2:M}]_{n'} } \nonumber \\
 +  \frac{\ds{1}_+(t_1-a_1)}{\prod_{n=1}^N B_{n,1}} \sum_{n=2}^{N-1} \biggl(\frac{ \exp( -\langle \bm{a}^{(n)}, \bg{\lambda}_{2:M} \rangle) }{ \prod_{n'=1}^{N-1} [\bm{B}^{(n)} \bg{\lambda}_{2:M}]_{n'} } \nonumber \\
 -  \frac{ \exp( -\langle \bm{a}^{(1)}, \bg{\lambda}_{2:M} \rangle) }{ \prod_{n'=1}^{N-1} [\bm{B}^{(n)} \bg{\lambda}_{2:M}]_{n'} } \biggr).
\end{align}

Setting $\bm{C} := \bm{B} \oslash (\bm{B}_{:,1} \bold{1}^T)$ shows that $\bm{a}^{(n)}\in \bb{R}^{M-1}$ and $\bm{B}^{(n)} \in \bb{R}^{N-1 \times M-1}$ ($n\in\{1,\hdots,N-1\})$ are equal to
\begin{align}
\bm{a}^{(n)} & = \begin{cases}
\bm{a}_{2:M} + (t_1 - a_1) \bm{C}_{1,2:M}, \ & n=1\\
\bm{a}_{2:M} + (t_1 - a_1) \bm{C}_{n+1,2:M}, \ &  n\geq 2
\end{cases}
\end{align}
\begin{align}
\bm{B}^{(n)} & =  \begin{cases}
\bm{C}_{3:N,2:M} - \bold{1} \bm{C}_{1,2:M}, \  & n=1 \\
\begin{bmatrix} \bm{C}_{n+1,2:M} - \bm{C}_{1,2:M} \\ \bm{C}_{3:N,2:M} - \bold{1} \bm{C}_{n+1,2:M} \end{bmatrix}, \  & n \geq 2.
\end{cases}
\end{align}

\end{proposition}
\begin{proof}
The proof is deferred to Appendix \ref{app:appD}.
\end{proof}
Using Prop. \ref{prop:exp_invlap_dpole} we can readily build a computational tree (resp. neural network) to compute the moment $\mu_k$ ($1\leq k \leq N$) comprising $\mc{O}(N^M)$ nodes as depicted in Fig. \ref{fig:computation_tree_moment}. It can be shown that the individual networks for volume and moments perform a set of identical computations, i.e., the networks share particular subnetworks, which results in a subadditive number of nodes for a combined network computing $\rs{vol}(\mc{P}_\bm{t})$ and $\bg{\mu}$. However, in the worst-case the resulting number of nodes in the network still grows as $\mc{O}(N^M)$.
\begin{remark}
The worst-case growth of $\mc{O}(N^M)$ corresponds to a fully connected network but in numerical experiments large subnetworks were never activated by the corresponding activation functions, regardless of the particular network input $\bm{t} \in \{ \bm{V}_s^T \bm{x}: \ \bm{x} \in \Delta\}$. In addition, the size of activated subnetworks strongly depends on the choice of the orthogonal basis vectors $\bm{V}_s$ of the subspace $\rs{ker}^{\perp}(\bm{A})$. The choice of a favorable (or even optimal) basis of the affine subspace is beyond the scope of this paper but poses an interesting question to be addressed in future works.
\end{remark}

\section{Numerical example: compressing soft-classification vectors}\label{sec:numresults}
To assess the performance of the proposed network, we consider the problem of \textit{soft-decision compression} for \textit{distributed decision-making}. In our setting, nodes reduce data traffic by transmitting only compressed versions of their local soft-classification vectors. At the receiver-side, a fusion center recovers the data by employing enhanced algorithms for robust decision results. One possible application of this scenario is a multi-view image classification that we consider in this section. More precisely, we study the compressibility of softmax-outputs of deep learning classifiers on MNIST handwritten digits and CIFAR-10 images obtained using MatConvNet (www.vlfeat.org/matconvnet/) to assess the fitness of the uniform simplex distribution for practical datasets. The outputs of the trained classifiers are vectors $\bt{x} \in \bb{R}_+^{10}$ that obey $\sum_{n=1}^{10} \tilde{x}_n = 1$, where each entry measures the estimated class-membership probability corresponding either to occurrence of digits $\{0,\hdots,9\}$ for MNIST, or occurrence of classes \textit{\{plane, car, bird, cat, deer, dog, frog, horse, ship, truck\}} for CIFAR-10. These vectors are preprocessed by removing one uninformative component (e.g. $\tilde{x}_{10} = 1 - \sum_{n=1}^{9} \tilde{x}_n$) which ensures that
\begin{align}
\bm{x}:= \bt{x}_{1:9} \in \Delta.
\end{align}
Examples of a realization $\bm{x}$ of $\bs{x}\sim\mc{U}(\Delta)$ as well as input images and outputs of standard classifiers are given in Fig. \ref{fig:mnist_examples}, \ref{fig:cifar_examples} and \ref{fig:realizations}, respectively.
\begin{figure}
\centering 
  \includegraphics[width=0.8\linewidth]{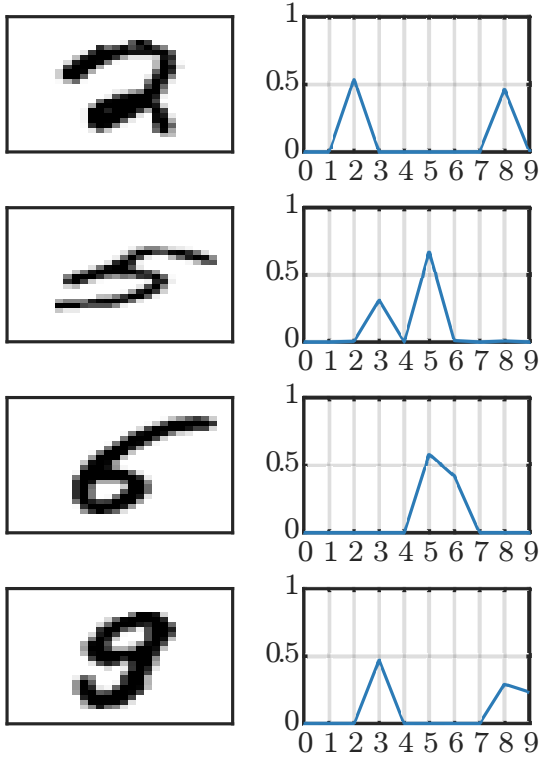}
\caption{MNIST images and softmax-classifications for low-confidence examples. True labels are $\{2,5,6,9\}$ and estimated labels are $\{2,5,5,3\}$.}
\label{fig:mnist_examples}
\end{figure}
We adjust the confidence levels to match the measured accuracy via the \textit{temperature}-parameter of the \textit{softmax}-output such that the top-entry is on average $0.9748$ for MNIST and $0.7987$ for CIFAR-10 (for the default parameter almost all decisions are made with unduly high confidence levels). We measure the empirical mean-square error 
\begin{align}
\rs{eMSE} := \frac{1}{N_s} \sum_{i=1}^{N_s} \lVert \bm{x}^{(i)} - \bh{x}^{(i)} \rVert_2^2
\end{align}
over a testing set of cardinality $N_s=500$ and compare the proposed estimator $\bh{x}$ \eqref{equ:centest} with exact centroid computation by simplicial decomposition using \textit{Qhull} \cite{BueEngFuk00} as well as solutions to the well-established non-negative \textit{$\ell_1$-minimization}
\begin{align}
\bh{x}_{\ell_1} = \underset{\bm{x}\in\bb{R}_+^N:\ \bm{y}=\bm{A}\bm{x}}{\rs{argmin}} \ \ \lVert \bm{x} \lVert_1 
\end{align}
and (simplex constrained) \textit{$\ell_2$-minimization}
\begin{align}
\bh{x}_{\ell_2} = \underset{\bm{x}\in\Delta:\ \bm{y}=\bm{A}\bm{x}}{\rs{argmin}} \ \ \lVert \bm{x} \lVert_2^2.
\end{align}
\begin{figure}
\centering 
  \includegraphics[width=0.8\linewidth]{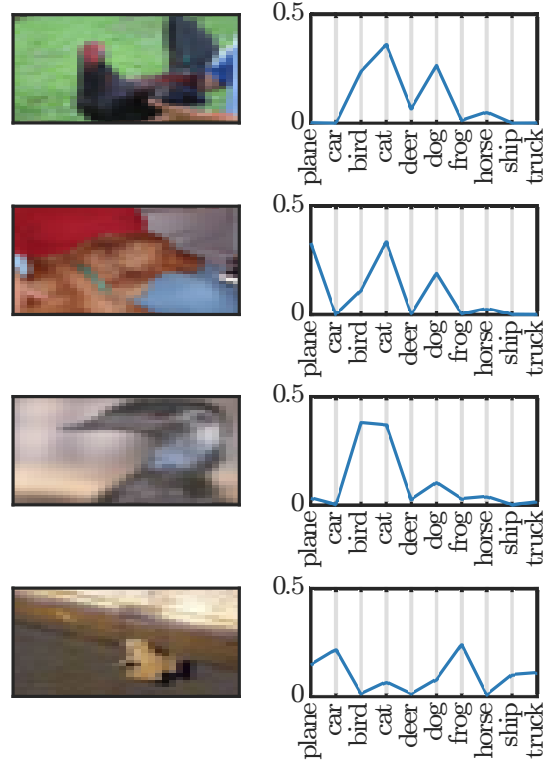}
\caption{CIFAR-10 images and softmax-classifications for low-confidence examples. True labels are \textit{\{bird, dog, bird, frog\}} and estimated labels are \textit{\{cat, cat, bird, frog\}}.}
\label{fig:cifar_examples}
\end{figure}

For the compression matrix $\bm{A}$, we use an i.i.d. random Gaussian matrix drawn once and set fixed for all simulations. The compressed soft-classification vector is given by $\bm{y} = \bm{A} \bm{x} \in \bb{R}^M$ where we vary the number of compressed measurements $M$. As a reference, we also showcase the results for input signals following the presumed uniform simplex distribution in Fig. \ref{fig:results_rnd}. 
It is interesting to see that for the prescribed uniform distribution the proposed centroid estimator outperforms the conventional $\ell_1$- and $\ell_2$-based methods by a factor of about $3$ and $2$, respectively (see Fig. \ref{fig:results_rnd}). For the MNIST and CIFAR-10 dataset, the proposed method is on par with the well-established $\ell_1$-minimization method (see Fig. \ref{fig:results_mnist} and \ref{fig:results_cifar}). All simulations were run on a laptop with i7-2.9 GHz processor.\footnote{In the spirit of reproducible research, the simulation code for the computation of volumes and centroids using Laplace techniques and simplicial decomposition is made available at \url{https://github.com/stli/CentNet}.}
Of course, additional performance gains of the proposed network can be expected from fine-tuning of the analytically obtained parameters based on given datasets.
As a side note, the volumes of the intersection polytopes can be surprisingly small and in some cases were below numerical precision causing precision problems and numerical underflows. On the other hand, numerical instability is a well-known problem in the design of deep neural networks (see e.g. \cite{GooBenCou16}) and appropriate numerical stabilization techniques, e.g., via logarithmic calculus, are often required. For the datasets at hand, numerical underflows occurred for a small number of MNIST examples, where the input $\bm{x}$ was close to a vertex $\bm{s}_i$ of the simplex $\Delta$ and the intersection volume becomes extremely small. We note that these examples were removed for the evaluation of the neural network but kept for all other estimators. As this case is rather easy to solve using conventional $\ell_1$-minimization and resulting estimation errors are typical smaller than average the depicted results do not favour the proposed approach. A detailed numerical analysis and numerically redesigned network is beyond the scope of this paper.

\begin{figure}
\centering 
  \includegraphics[width=1\linewidth]{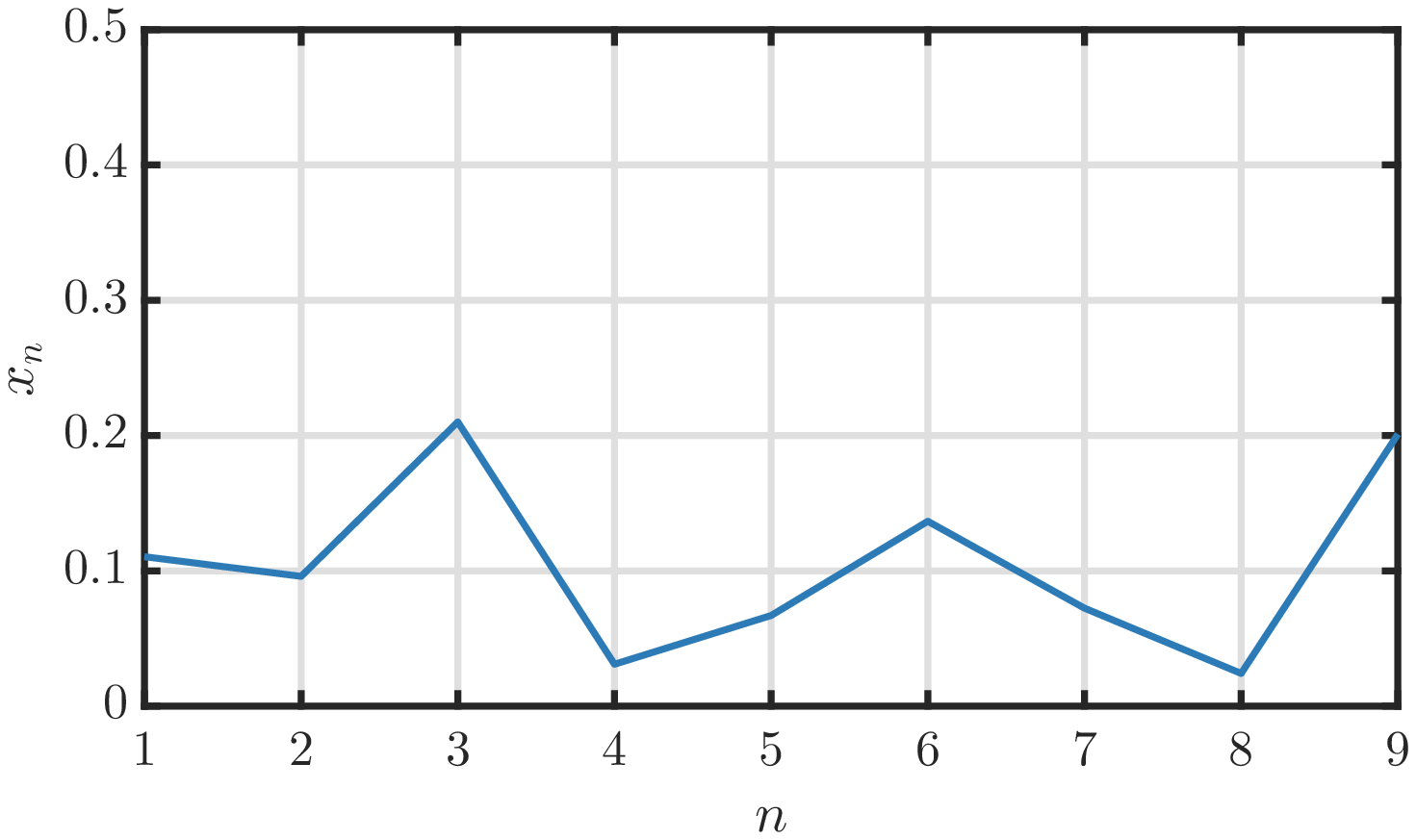}
\caption{Realization of $\bs{x} \sim \mc{U}(\Delta)$ for $N=9$.}
\label{fig:realizations}
\end{figure}

\begin{figure}
\centering 
  \includegraphics[width=1\linewidth]{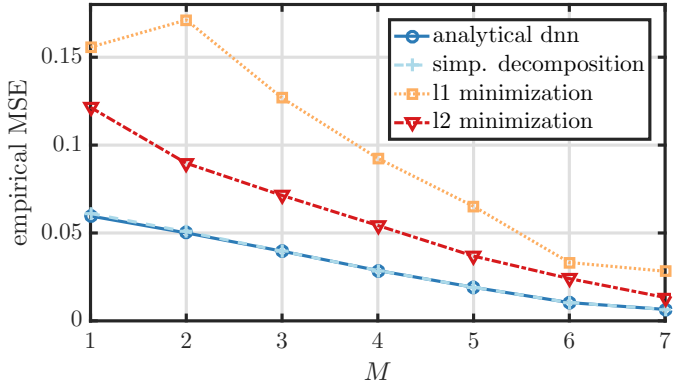}
\caption{Empirical MSE for $N=9$, $\bm{x}\sim\mc{U}(\Delta)$, i.i.d. Gaussian matrix $\bm{A}$ and varying number of measurements $M$.}
\label{fig:results_rnd}
\end{figure}
\begin{figure}
\centering 
  \includegraphics[width=1\linewidth]{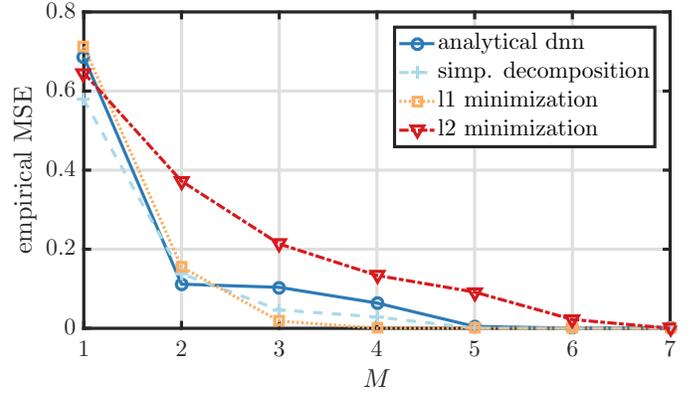}
\caption{Empirical MSE for MNIST dataset, i.i.d. Gaussian matrix $\bm{A}$ and varying number of measurements $M$.}

\vspace{-9pt}
\label{fig:results_mnist}
\end{figure}
\begin{figure}
\centering 
  \includegraphics[width=1\linewidth]{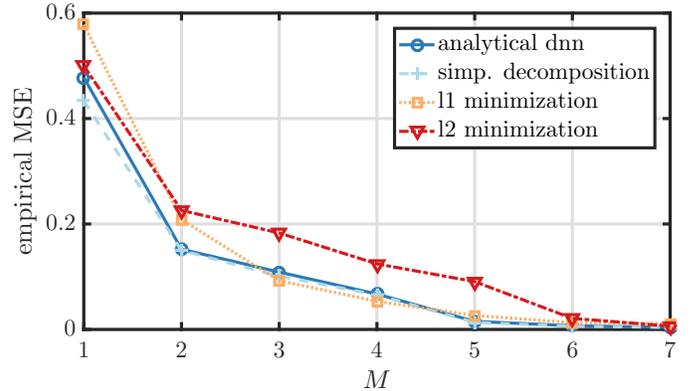}
\caption{Empirical MSE for CIFAR-10 dataset, i.i.d. Gaussian matrix $\bm{A}$ and varying number of measurements $M$.}
\label{fig:results_cifar}
\end{figure}

\section{Conclusion}
In this paper we proposed a novel theoretically well-founded neural network for sparse recovery. By using multidimensional Laplace techniques and a prescribed input distribution, we obtain a neural network in a fully analytical fashion. Interestingly, the obtained neural network is composed of weights as well as commonly employed threshold functions, rectified linear (ReLU) and rectified polynomial (ReP) activation functions. The obtained network is a first step to understanding the practical effectiveness of classical deep neural network architectures. To scale to higher dimensions, a main problem is to decrease the network width which may be achieved by deactivating maximally large subnetworks via a well-chosen basis of the affine subspace which poses an interesting problem for future works. In addition, it may be beneficial to investigate approximations of the constructed network by a smaller subnetwork which may yield a reasonable approximation of the centroid of interest.

\appendices
\section{Proof of Proposition \ref{prop:volPt}}\label{app:appA}
First note that $\mc{P}_\bm{t}$ is an $(N-M)$-dimensional face of $\mc{T}_\bm{t}$. Hence, repeated application of \cite[Prop. 3.3]{Las83} shows that

\begin{align}\label{equ:polyface_differentiate}
\rs{vol}(\mc{P}_\bm{t}) = \lVert \bm{v}_1 \rVert_2 \cdots \lVert \bm{v}_M \rVert_2 \frac{\partial^M}{\partial t_1 \cdots \partial t_M} \rs{vol}(\mc{T}_\bm{t}).
\end{align}

Considering \eqref{equ:lapid} with $f(\bm{t})=\rs{vol}(\mc{P}_\bm{t})$ and the fact that $\lVert \bm{v}_m \rVert_2 =1$ yields
\begin{align}
\rs{vol}(\mc{P}_\bm{t}) & = \left( \circ_{m} \mc{L}_{\lambda_m}^{-1} \right)\left( \circ_{m} \mc{L}_{\lambda_m} \right) \rs{vol}(\mc{P}_\bm{t}) \nonumber \\
& = \left( \circ_{m} \mc{L}_{\lambda_m}^{-1} \right)\left( \circ_{m} \mc{L}_{\lambda_m} \right) \frac{\partial^M}{\partial t_1 \cdots \partial t_M} \rs{vol}(\mc{T}_\bm{t}).
\end{align}

By continuity of $\rs{vol}(\mc{T}_\bm{t})$ we have $\forall \ m \in \{1,\hdots,M\}$ that
\begin{align}
\rs{lim}_{t_m \to 0} \ \rs{vol}(\mc{T}_\bm{t}) = 0.
\end{align}
So repeated application of the transform pair (LT7) yields the desired result.

\section{Proof of Lemma \ref{lem:exp_invlap}}\label{app:appB}
For Lemma \ref{lem:exp_invlapM1} with $M=1$ the stated transform pair is obtained by using the transform pairs (LT6), (LT2) and linearity (LT1).

For Lemma \ref{lem:exp_invlapM2} with $M\geq2$ we first prove the transform pair ($\lambda_n^{(0)}\neq \lambda_{n'}^{(0)}$, $n \neq n'$)
\begin{align}\label{equ:appB_pair1}
\frac{c  \exp(-a \lambda) }{ \prod_{n=1}^N (\lambda - \lambda_n^{(0)}) } \ \overset{ a \geq 0 }{ \Laplace } \ \ds{1}_+(t-a) \sum_{n=1}^N \frac{ c \exp( \lambda_n^{(0)} (t-a) ) }{ \prod_{n'\neq n} ( \lambda_{n}^{(0)} - \lambda_{n'}^{(0)} ) }.
\end{align}
To this end, we use the partial fraction expansion (see \cite[(10) p.77]{Doe12})
\begin{align}\label{equ:partfrac}
\frac{1}{\prod_{n=1}^N ( \lambda - \lambda_n^{(0)} )} & = \sum_{n=1}^N \frac{1}{ \lambda - \lambda_n^{(0)}} \frac{1}{ \frac{d}{d \lambda} \prod_{n'=1}^N ( \lambda - \lambda_{n'}^{(0)} ) \vert_{\lambda = \lambda_n^{(0)} } } \nonumber \\
& = \sum_{n=1}^{N} \frac{1}{ \lambda - \lambda_n^{(0)}} \frac{ 1 }{\prod_{n'\neq n} (\lambda_{n}^{(0)} - \lambda_{n'}^{(0)} ) }
\end{align}
in conjunction with the transform pair (LT3) to obtain the transform pair
\begin{align}\label{equ:appB_pair2}
\frac{1}{ \prod_{n=1}^N (\lambda - \lambda_n^{(0)} )} \ \Laplace \ \sum_{n=1}^{N} \frac{ \exp( \lambda_n^{(0)} t) }{\prod_{n'\neq n} (\lambda_{n}^{(0)} - \lambda_{n'}^{(0)} ) }.
\end{align}
Then, \eqref{equ:appB_pair1} follows from \eqref{equ:appB_pair2} by linearity and the transform pair (LT6). Finally, by assumption,
\begin{align}
\lambda_n^{(0)} := -\frac{1}{B_{n,1}} \bm{B}_{n,2:M} \bg{\lambda}_{2:M}
\end{align}
are pairwise distinct (we can assume $\bg{\lambda}_{2:M}$ is arbitrary but fixed) and the result \eqref{equ:ilt_exppoly_M2} follows from \eqref{equ:appB_pair1} by setting $\lambda:=\lambda_1$, $a:=a_1$ and $c:= (\prod_{n=1}^N B_{n,1})^{-1} \exp( - \langle \bm{a}_{2:M}, \bg{\lambda}_{2:M} \rangle)$.

\section{Proof of Lemma \ref{lem:baldoni_extended}}\label{app:appC}
To obtain the desired integral, we assume $\bm{l}_{1:N \setminus k}$ is arbitrary but fixed, and consider the function 
\begin{align}
f(\bm{x},l_k):=\exp(-\langle \bm{l} , \bm{x} \rangle).
\end{align}
As $f$ is jointly continuous in the variables $\bm{x}$, $l_k$ and $\frac{\partial}{\partial l_k} f(\bm{x},l_k)$ is continuous, it holds that (see \cite[Th. 8.11.2]{Die69})
\begin{align*}
& \int_\Delta x_k \exp(-\langle \bm{l}, \bm{x} \rangle ) \ d\bm{x} = \int_{\Delta} -\frac{ \partial }{\partial l_k} f(\bm{x},l_k) \ d\bm{x}  \\
&= - \frac{\partial }{ \partial l_k } \int_\Delta f(\bm{x},l_k) \ d\bm{x}
= - \frac{\partial }{ \partial l_k } \sum_{n=1}^{N+1} \frac{ \exp( - \langle \bm{l}, \bm{s}_n \rangle ) }{ \prod_{n'\neq n} \langle \bm{l}, \bm{s}_{n'} - \bm{s}_n \rangle}.
\end{align*}
Carrying out the differentiation yields the desired result.

\section{Proof of Proposition \ref{prop:exp_invlap_dpole}}\label{app:appD}
First we obtain the transform pair $F(\lambda)\Laplace f(t)$ ($\lambda_n^{(0)}\neq \lambda_{n'}^{(0)}$ for $n \neq n'\in\{1,\hdots,N-1\}$)
\begin{align}\label{equ:dpole_eq1}
\frac{\exp(a \lambda)}{(\lambda-\lambda_1^{(0)}) \prod_{n=1}^{N-1} (\lambda - \lambda_n^{(0)})}  \overset{ a \geq 0 }{ \Laplace }  
 \frac{ (t-a)_+ \exp( \lambda_1^{(0)}(t-a)) }{\prod_{n=2}^{N-1} (\lambda_1^{(0)} - \lambda_n^{(0)}) } \nonumber \\
 + \ds{1}_+(t-a) \sum_{n=2}^{N-1} \frac{ \exp(\lambda_n^{(0)}(t-a)) - \exp(\lambda_1^{(0)}(t-a))}{ (\lambda_n^{(0)} - \lambda_1^{(0)}) 
\prod_{n'\neq n} (\lambda_n^{(0)} - \lambda_{n'}^{(0)}) }
\end{align}
by using the partial fraction expansion \eqref{equ:partfrac} 
\begin{align}\label{equ:partfrac}
\frac{1}{( \lambda - \lambda_1^{(0)} ) \prod_{n=1}^{N-1} ( \lambda - \lambda_n^{(0)} )} = \\ = \frac{1}{( \lambda - \lambda_1^{(0)} )}\sum_{n=1}^{N-1} \frac{1}{ \lambda - \lambda_n^{(0)}} \frac{ 1 }{\prod_{n'\neq n} (\lambda_{n}^{(0)} - \lambda_{n'}^{(0)} ) }
\end{align}
in conjunction with the transform pairs (LT4), (LT5) and (LT6).
By assumption, $\bm{B}_{1,:} = \bm{B}_{2,:}$ and $\bg{\lambda}_{2:M}$ is arbitrary but fixed so that
\begin{align}
\lambda_n^{(0)} := \begin{cases}
-\frac{1}{B_{1,1}} \bm{B}_{1,2:M} \bg{\lambda}_{2:M}, \ &n=1 \\
-\frac{1}{B_{n,1}} \bm{B}_{n+1,2:M} \bg{\lambda}_{2:M}, \ &n\in\{2,\hdots,N-1\}
\end{cases}
\end{align}
are pairwise distinct and \eqref{equ:exp_invlap_M2_dpole} follows from \eqref{equ:dpole_eq1} by setting $\lambda:=\lambda_1$, $a:=a_1$ and multiplying both sides by $c:=(\prod_{n=1}^N B_{n,1})^{-1} \exp(-\langle \bm{a}_{2:M}, \bg{\lambda}_{2:M}\rangle)$.



\ifCLASSOPTIONcaptionsoff
  \newpage
\fi

\IEEEtriggeratref{15}
\bibliographystyle{IEEEbib}
\bibliography{refs}

%

\end{document}